\documentclass[12pt,reqno]{article}

\usepackage[usenames]{color}
\usepackage{amssymb}
\usepackage{amsmath}
\usepackage{amsthm}
\usepackage{amsfonts}
\usepackage{amscd}
\usepackage{graphicx}
\usepackage{xcolor}
\usepackage{float}

\usepackage[colorlinks=true,
linkcolor=webgreen,
filecolor=webbrown,
citecolor=webgreen]{hyperref}

\definecolor{webgreen}{rgb}{0,.5,0}
\definecolor{webbrown}{rgb}{.6,0,0}

\usepackage{fullpage}

\usepackage{graphics}
\usepackage{latexsym}
\usepackage{epsf}
\usepackage{breakurl}

\begin{document}

\def\suchthat{\, : \,}
\def\modd#1 #2{#1\ \mbox{\rm (mod}\ #2\mbox{\rm )}}

\title{The speed of convergence in the Cooper-Dutle dueling game}

\author{Jeffrey Shallit\\
School of Computer Science\\
University of Waterloo\\
Waterloo, ON N2L 3G1 \\
Canada \\
\href{mailto:shallit@uwaterloo.ca}{\tt shallit@uwaterloo.ca}}

\maketitle

\theoremstyle{plain}
\newtheorem{theorem}{Theorem}
\newtheorem{corollary}[theorem]{Corollary}
\newtheorem{lemma}[theorem]{Lemma}
\newtheorem{proposition}[theorem]{Proposition}

\theoremstyle{definition}
\newtheorem{definition}[theorem]{Definition}
\newtheorem{example}[theorem]{Example}
\newtheorem{conjecture}[theorem]{Conjecture}

\theoremstyle{remark}
\newtheorem{remark}[theorem]{Remark}

\def\Enn{\mathbb{N}}

\begin{abstract}
In 2013 Cooper and Dutle invented a dueling scenario where 
Alice and Bob shoot at each other until one is hit.  Each shot is
successful with some fixed probability $p$, $0 < p < 1$.
The shooting order is given by a greedy algorithm, where at
each step a shot is assigned to the player whose current probability
of success is smaller.

Cooper and Dutle observed that as $p \rightarrow 0$, the
resulting sequence specifying who shoots at the
$n$th turn (Alice or Bob) converges to the infinite Thue-Morse
sequence $\bf t$. But they left the speed of convergence as an
open problem.
In this note we determine the speed of this convergence.
\end{abstract}

\vspace{10px}

\noindent{\textbf{Keywords:}  {probability, Thue-Morse sequence, greedy algorithm.}}

\vspace{10px}

\section{Introduction}

The {\it Thue-Morse\/} sequence ${\bf t} = (t(i))_{i \geq 0}$ is a
celebrated infinite
sequence of $+1$ and $-1$ defined by $t(i) = (-1)^{s_2 (i)}$, where
$s_2 (i)$ is the sum of the bits of the binary representation of $i$.
It often makes its appearance in surprising ways
and in unexpected places; see the survey
\cite{Allouche&Shallit:1999} for some examples.\footnote{In
the literature, the Thue-Morse
sequence is often defined with terms $\{0,1 \}$ or $\{a,b\}$
instead of $\{1,-1\}$ as we do here.}
The first few terms are given in Table~\ref{tab1}.
From the definition, we immediately see that 
\begin{equation}
t(2i) = t(i) \quad \text{and} \quad t(2i+1) = -t(i) 
\label{tmpair}
\end{equation}
for $i \geq 0$.

\begin{table}[h]
\begin{center}
\caption{First few values of the Thue-Morse sequence.}
\begin{tabular}{c|cccccccccccccccccccccccccccccccc}
$i$ & 0& 1& 2& 3& 4& 5& 6& 7& 8& 9&10&11&12&13 \\
\hline
$t(i)$ &  1&$-1$&$-1$& 1&$-1$& 1& 1&$-1$&$-1$& 1& 1&$-1$& 1&$-1$
\end{tabular}
\end{center}
\label{tab1}
\end{table}

In a {\it Monthly} article published in
2013, Joshua Cooper and Aaron Dutle found yet another place where
the Thue-Morse sequence makes a natural appearance.
They considered the following scenario \cite{Cooper&Dutle:2013}.
Alice and Bob participate in a duel;
the game ends when one person---the winner---hits the other.
When a shot is made by either player, it causes a hit
with some fixed probability $p$, for some $p$ with  $0 < p < 1$.

Before the game begins, they agree to take shots in a fixed
order, based on a certain greedy fairness rule, as follows:  Alice
goes first, at step $0$.  At step $n$,
they compute the prior probability $P_A(n)$ that Alice would have won
before step $n$, and compare it to the probability 
$P_B(n)$ that Bob would have won before step $n$.  If $P_A(n) > P_B(n)$,
then Bob gets to shoot at step $n$; otherwise Alice does.

For example, if $p = 1/3$, then Alice shoots first.  Now $P_A(1) = 1/3$, and
since Bob couldn't have won without shooting, we have $P_B(1) = 0$.
Since $P_A(1) > P_B(1)$, Bob shoots at step $1$.   Now $P_A(2)$ is still
$1/3$, but $P_B(2) = 2/3$ (the probability Alice missed in step $0$) times
$1/3$ (the probability Bob hits at step $1$), equalling $2/9$.
Since $1/3 > 2/9$,
Bob gets to shoot at step $2$.  The resulting sequence of shots
begins Alice, Bob, Bob, Alice, Bob, Alice, Bob, $\ldots$.

Cooper and Dutle observed that the analysis becomes simpler
if it is phrased in terms of $q = 1-p$, and we follow their
convention.
If we encode Alice's shot as $+1$ and Bob's as $-1$, then
Cooper and Dutle \cite[Eq.~(1), p.~442]{Cooper&Dutle:2013} proved
the following.  
\begin{proposition}
Let $q = 1-p$.
For $s_0, s_1, \ldots, s_{n-1}$ an arbitrary (not necessarily
greedily chosen) finite sequence of shots
written as $\pm 1$, the difference $P_A(n) - P_B(n)$ between Alice's
and Bob's winning probabilities  is
$ p \sum_{0 \leq i <n} s_i q^i$.
\label{lem01}
\end{proposition}

\begin{proof}
If $s_i = 1$, then Alice gains an extra amount of probability
from her shot at that turn, which is $p$ (the probability of
her shot's success) times $q^i$ (the probability that all
previous shots failed).  The same argument applies to Bob
for $s_i = -1$.  So 
$$P_A(n) = p \sum_{{0 \leq i < n}\atop {s_i = 1}} q^i, \quad\quad
P_B(n) = p \sum_{{0 \leq i < n}\atop {s_i = -1}} q^i,$$
giving the formula for $P_A(n)-P_B(n)$. 
\end{proof}
\begin{corollary}
Let $g_q (i)$ for $i \geq 0$ be chosen according to the greedy
rule, and define $G_q(n) = \sum_{0 \leq i < n} g_q(i) q^i$.
Then $g_q(n)$ is $-1$ if 
$G_q(n)$ is positive, and $+1$ if it is negative or zero.
\label{cor02}
\end{corollary}

Cooper and Dutle proved that 
as $q \rightarrow 1$, the
sequence $(g_q(i))_{i \geq 0}$ agrees with the Thue-Morse sequence
${\bf t} = (t(i))_{i \geq 0}$ on longer and longer prefixes of
$\bf t$.

For example, for $q = 2/3$ the sequence $(g_q(i))_{i \geq 0}$
begins 
$$1, -1, -1, 1, -1, 1, -1,\ldots,$$
which agrees with $(t(i))_{i \geq 0}$ on the first $6$ terms, but
then differs at the $7$th term.
This suggests the following natural question:  as a function of $q$,
how long does the shooting sequence $(g_q(i))_{i \geq 0}$  agree with $\bf t$?
In this note we answer this question.

In retrospect,  it is not that surprising that the Thue-Morse
sequence might make an appearance in a scenario like this,
because the number of $+1$s in every prefix never differs from the
number of $-1$s in that prefix by more than $1$.
Thue-Morse appears in many other situations involving fair
allocation (see \cite{Barrow:2010,Brams&Ismail:2018,Palacios:2012} and
\cite[Chap.~3]{Brams&Taylor:1999})
and greedy choice \cite{Allouche&Cohen:1985}.

\section{The main result}

Let $L_q$ be the length of the longest prefix where
$(g_q(i))_{i \geq 0}$ agrees with $\bf t$.
A little numerical experimentation, as illustrated in
Figure~\ref{fig1}, suggests that $L_q$ always belongs
to the integer sequence
\begin{equation}
3, 5,6,10,12, 20,24,40,48,80,96,\ldots,
\label{values}
\end{equation}
which are the numbers of the form $3 \cdot 2^i$ and
$5 \cdot 2^i$ for $i \geq 0$.  \footnote{This sequence, missing its first
term $3$, is sequence A164095 in the On-Line
Encyclopedia of Integer Sequences \cite{oeis}.}

\begin{figure}
\begin{center}
\includegraphics[width=6.0in]{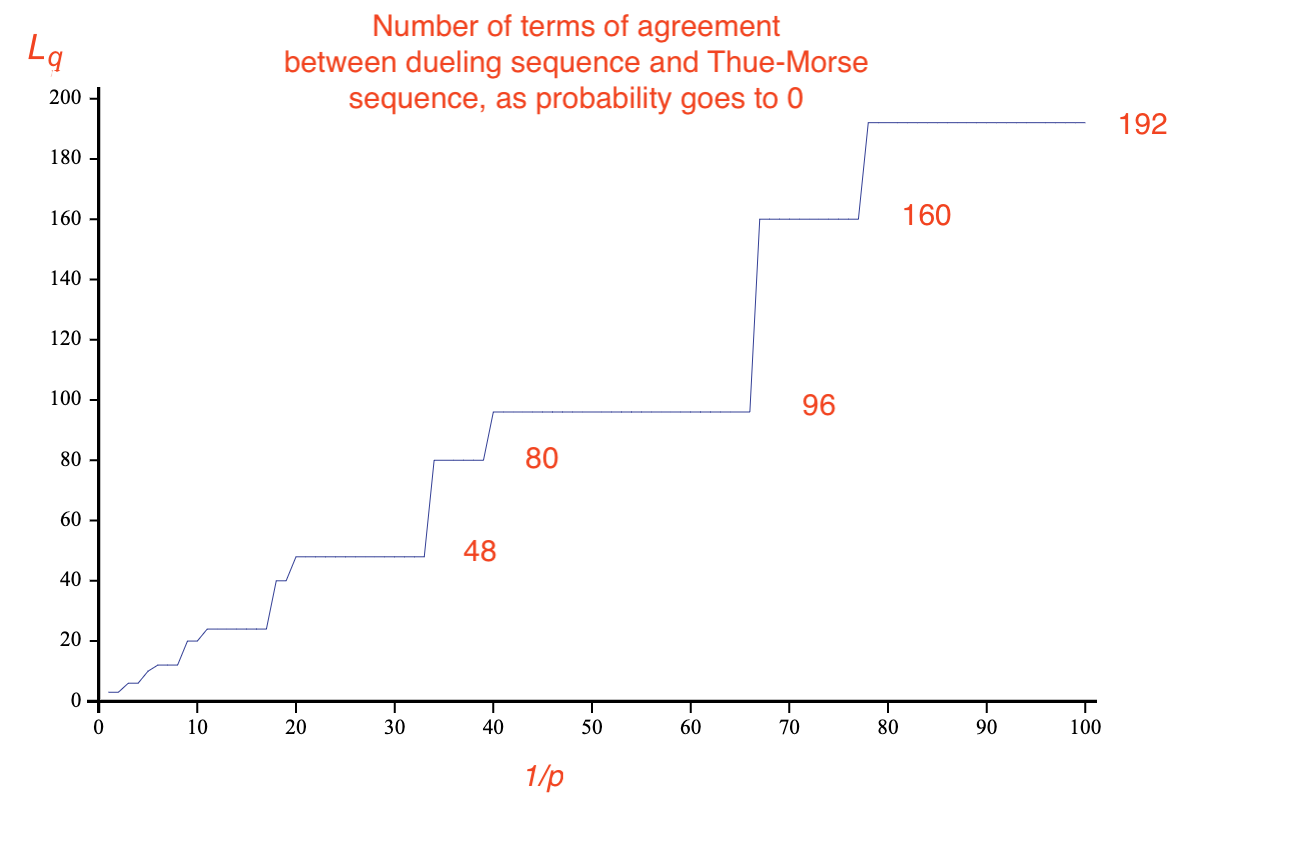}
\end{center}
\caption{$L_q$ on the $y$-axis, $1/p = 1/(1-q)$ on the $x$-axis.}
\label{fig1}
\end{figure}

More precisely, here is our main result.

\begin{theorem} Let
$\alpha = (\sqrt{5}-1)/2$ 
be the positive zero of $1-X-X^2$, and
$\beta$ be the positive real zero of
$h(X) := 1-X-X^2+X^3-X^4$.  Note that $\alpha
\approx 0.61803398875$ and
$\beta \approx 0.6609925319$.
Then
$$ L_q = \begin{cases}
3, & \text{if \, $0 < q < \alpha$;} \\
5 \cdot 2^n, & \text{if\, $\alpha^{2^{-n}} \leq q < \beta^{2^{-n}}$
	for $n \geq 0$;} \\
3 \cdot 2^{n+1}, & \text{if \, $\beta^{2^{-n}} \leq q < 
	\alpha^{2^{-n-1}}$ for $n\geq 0$ .}
\end{cases}
$$
\label{mainn}
\end{theorem}

\begin{lemma}\label{lem:tail}
If $\tfrac12\le q<1$, then
\begin{equation}\label{eq:tail}
  |G_q(n)|\le\frac{q^n}{1-q}
\end{equation}
for $n \geq 0$.
\end{lemma}

\begin{proof}
When going from $G_q(n)$ to $G_q(n+1)$, the
greedy rule subtracts $q^n$ if $G_q(n) > 0$, and
adds $q^n$ if $G_q(n) < 0$.
Taking absolute values, we see that
\[
  |G_{q}(n+1)|=\bigl||G_q(n)|-q^n\bigr|.
\]
This identity also holds when $G_q(n)=0$:
both sides of the identity are then $q^n$.

We now prove the bound \eqref{eq:tail} by induction on $n$. It holds for $n=0$,
since $G_q(0)=0$. Suppose it holds for $n$, and put
\[
  u=|G_q(n)|,\qquad w=q^n,\qquad B=\frac{q^n}{1-q}.
\]
The induction hypothesis says $0\le u\le B$. Subtracting $w$
gives
\[
  -w\le u-w\le B-w.
\]
If $u\le w$, then $|u-w|=w-u\le w$; if $u\ge w$, then
$|u-w|=u-w\le B-w$. Thus
$ |u-w|\le\max\{w,B-w\}$.
Substituting the definitions and using the absolute-value
recurrence, we obtain
\begin{align*}
  |G_{q}(n+1)|
  &\le\max\left\{q^n,\frac{q^n}{1-q}-q^n\right\}\\
  &=\max\left\{q^n,\frac{q^{n+1}}{1-q}\right\}\\
  &=\frac{q^{n+1}}{1-q}.
\end{align*}
The last equality follows from
\[
  \frac{q^{n+1}}{1-q}\ge q^n
  \quad\Longleftrightarrow\quad
  \frac{q}{1-q}\ge1
  \quad\Longleftrightarrow\quad
  q\ge 1/2.
\]
The induction is complete.
\end{proof}

\begin{lemma}\label{lem:pairing}
If $0<q<1$ and $q^2\ge\tfrac12$, then
\begin{equation}\label{eq:pairing}
  g_{q}(2n)=g_{q^2}(n),
  \qquad
  g_{q}(2n+1)=-g_{q^2}(n)
\end{equation}
for all $n \geq 0$.
Consequently,
\begin{equation}\label{eq:doubling}
  L_q=2L_{q^2}.
\end{equation}
\end{lemma}

\begin{proof}
Fix $q$ satisfying the hypotheses, and put $b_j=g_{q^2}(j)$.
For each $n\ge0$, let $P(n)$ be the assertion
\begin{equation}\label{eq:pair-induction}
  g_q(2j)=b_j,\qquad g_q(2j+1)=-b_j
  \qquad(0\le j<n).
\end{equation}
We prove $P(n)$ by induction on $n$. The assertion $P(0)$ is
vacuously true.

Suppose $P(n)$ holds.  Then
\begin{align*}
  G_{q}(2n)
  &=\sum_{j=0}^{n-1}
    \bigl(g_q(2j)q^{2j}+g_q(2j+1)q^{2j+1}\bigr)\\
  &=\sum_{j=0}^{n-1}\bigl(b_jq^{2j}-b_jq^{2j+1}\bigr) \quad \text{(by induction)}\\
  &=(1-q)\sum_{j=0}^{n-1}b_j (q^2)^j\\
  &=(1-q)G_{q^2}(n).
\end{align*}
Since $1-q>0$, the quantities $G_{q}(2n)$ and $G_{q^2}(n)$ have
the same sign and vanish together. Applying the greedy rule at
the two parameters therefore gives
\[
  g_q(2n)=g_{q^2}(n)=b_n.
\]

We now determine the shot at index $2n+1$. Since $q^2\ge 1/2$,
Lemma~\ref{lem:tail} applies at parameter $q^2$ and gives
\begin{align*}
  |G_{q}(2n)|
  &=(1-q)|G_{q^2}(n)|\\
  &\le\frac{(1-q)q^{2n}}{1-q^2}\\
  &=\frac{q^{2n}}{1+q}
  <q^{2n}.
\end{align*}
After the shot at index $2n$, we have
\[
  G_{q}(2n+1)=G_{q}(2n)+b_nq^{2n}.
\]
Multiplying by $b_n$ and using $b_n^2=1$, we obtain
\begin{align*}
  b_nG_{q}(2n+1)
  &=b_nG_{q}(2n)+q^{2n}\\
  &\ge q^{2n}-|G_{q}(2n)|\\
  &>0.
\end{align*}
Thus $G_{q}(2n+1)$ strictly has sign $b_n$. If $b_n=+1$, then
$G_q(2n+1)$
is positive and the next shot is $-1$; if $b_n=-1$,
$G_q(2n+1)$ is negative and the next shot is $+1$. In either
case,
\[
  g_q(2n+1)=-b_n.
\]
This proves both assertions at $j=n$, and hence proves $P(n+1)$.
The induction establishes \eqref{eq:pairing} for every $n$.

Comparing Eq.~\eqref{eq:pairing} with Eq.~\eqref{tmpair}, we 
get the desired result:  if the first disagreement of
$g_{q^2}(i)$ and $t(i)$ occurs when $i = n$, 
then the first disagreement of
$g_q (i)$ and $t(i)$ occurs when $i = 2n$, and vice versa.
\end{proof}

We can now prove our main result.

\begin{proof}[Proof of Theorem~\ref{mainn}]
First note that $h(0)=1$, $h(1)=-1$, and
\[
  g'(x)=-1-x(4x^2-3x+2)<0\qquad(x>0),
\]
because the quadratic $4x^2-3x+2$ is positive everywhere.
Thus $g$ has a unique positive zero $\beta$, lying in the
interval $(0,1)$.
Since $\alpha^2=1-\alpha$,
\[
  h(\alpha)=\alpha^3(1-\alpha)>0,
  \qquad
  h(\, \sqrt\alpha \,)=-\sqrt\alpha(1-\alpha)<0.
\]
It follows that
\begin{equation}\label{eq:ordering}
  \tfrac12<\alpha<\beta<\sqrt\alpha<1.
\end{equation}

We next calculate $L_q$ for $0<q<\sqrt\alpha$.
The first seven Thue--Morse terms are
\[
  1,-1,-1,1,-1,1,1.
\]
The first three greedy shots are always $1,-1,-1$, since
$G_q(0)=0$, $G_q(1)=1>0$, and $G_q(2)=1-q>0$.

\begin{enumerate}
\item If $0<q<\alpha$, then
\[
  G_q(3)=1-q-q^2>0.
\]
Thus $g_q(3)=-1\ne t(3)$, so $L_q=3$.

\item If $\alpha\le q<\beta$, then $G_q(3)\le0$, so
$g_q(3)=+1=t(3)$, including at $q=\alpha$. After this shot,
\[
  G_q(4)=(1-q)(1-q^2)>0,
\]
so $g_q(4)=-1=t(4)$. But $G_q(5)=h(q)>0$, giving
$g_q(5)=-1\ne t(5)$. Hence $L_q=5$.

\item If $\beta\le q<\sqrt\alpha$, the same calculations give
agreement through index $4$. Now $G_q(5)=h(q)\le0$, so
$g_q(5)=+1=t(5)$, including at $q=\beta$. Then
\[
  G_q(6)=(1-q)(1-q^2-q^4)>0,
\]
because $q^2<\alpha$. Consequently $g_q(6)=-1\ne t(6)$, and
$L_q=6$.
\end{enumerate}

We have established
\begin{equation}\label{eq:base}
  L_q=
  \begin{cases}
    3,&0<q<\alpha,\\
    5,&\alpha\le q<\beta,\\
    6,&\beta\le q<\sqrt\alpha.
  \end{cases}
\end{equation}

Now take any $\alpha\le q<1$. There is a unique integer $n\ge0$
such that
\begin{equation}\label{eq:reduction}
  \alpha\le u:=q^{2^n}<\sqrt\alpha.
\end{equation}
Indeed, choose the first $n$ for which $q^{2^n}<\sqrt\alpha$;
such an $n$ exists because $q^{2^n}\to0$. If $n=0$, the lower
bound follows from $q\ge\alpha$. If $n>0$, the preceding power
is at least $\sqrt\alpha$, so its square $u$ is at least
$\alpha$. Uniqueness follows because the next power satisfies
$u^2<\alpha$, and all subsequent powers are still smaller.

For every $j<n$, the parameter $q^{2^j}$ is at least
$\sqrt\alpha$, so its square is at least $\alpha>\tfrac12$.
Lemma~\ref{lem:pairing} therefore applies at each step, yielding
\[
  L_q=2^nL_u.
\]
If $\alpha\le u<\beta$, Eq.~\eqref{eq:base} gives
$L_q=5\cdot2^n$, and this condition is equivalent to
\[
  \alpha^{2^{-n}}\le q<\beta^{2^{-n}}.
\]
If $\beta\le u<\sqrt\alpha$, Eq.~\eqref{eq:base} gives
$L_q=6\cdot2^n=3\cdot2^{n+1}$, and this condition is equivalent to
\[
  \beta^{2^{-n}}\le q<\alpha^{2^{-n-1}}.
\]
Together with the first case of \eqref{eq:base}, these are all
the asserted cases.
\end{proof}

As a simple corollary, we get the behavior of $L_q$ when $p$ is
a power of $1/2$.
\begin{corollary}
For $p = 2^{-k}$ and $k \geq 1$,
the sequence $(g_q(i))_{i \geq 0}$ agrees with
$(t(i))_{i \geq 0}$ for the first $3 \cdot 2^{k-1}$ terms, and then
disagrees.
\label{cor}
\end{corollary}

\begin{proof}
By Theorem~\ref{mainn}, it 
suffices to show that  $\beta^{2^{2-k}} < 1-2^{-k} < \alpha^{2^{1-k}}$.

We prove the first inequality by induction on $k$.  We have
$\beta < 2/3$, so $\beta^2 < 4/9 < 1/2$, which proves the case $k = 1$.
Otherwise, assume $\beta^{2^{2-k}} < 1-2^{-k}$ and use
the inequality $\sqrt{1-x} \leq 1-x/2$.  We get
$\beta^{2^{2-(k+1)}} < \sqrt{1-2^{-k}} \leq  1-2^{-(k+1)}$, as desired.

We prove the second inequality by observing that
$\log(1-x) <-x $ for $x \in (0,1)$ (a consequence of  Taylor's theorem). 
Now $2 \log\alpha \approx -0.9624 > -1$, so
$\log(1-2^{-k}) < -2^{-k} <  2^{1-k} \log \alpha$, which
gives $1-2^{-k} < \alpha^{2^{1-k}}$.
\end{proof}

\section{How this paper came to be}

The author JS gave a talk on the Thue-Morse sequence
for the Prison Mathematics Project
in March 2022 \cite{Shallit:2022}.  As part of preparing this talk, JS
was led to consider the results of Cooper and Dutle.  Their Conjecture
1.2 seemed tractable, and numerical experiments suggested
the values in
Eq.~\eqref{values} and Corollary~\ref{cor}.
Despite an offer of US \$25 to anyone who could
prove the conjecture, no one did.

On April 28 2026 JS asked ChatGPT-5.5 Thinking if it
could prove the results \cite{ChatGPT1}.  First, it succeeded on the weaker result
of Corollary~\ref{cor}.  (The proof presented here is different.)
Then, after some suggestions by JS, it gave a proof of
Theorem~\ref{mainn}, but it had a subtle flaw.  Later discussion with
ChatGPT-5.6 Astra resolved the flaw \cite{ChatGPT2}.

I thank the referees for their comments.

\end{document}